\documentclass[runningheads]{article}
\usepackage[T1]{fontenc}
\usepackage{complexity}
\usepackage{comment}
\usepackage{todonotes}
\usepackage{xspace}
\usepackage[hidelinks]{hyperref}
\usepackage{enumitem}
\usepackage{subcaption}
\usepackage{amsmath}
\usepackage{stmaryrd}
\usepackage{amssymb}
\usepackage{amsthm}
\usepackage{authblk}

\usepackage{graphicx}
\usepackage{color}

\newcommand{\strat}{\mathcal{S}}
\newcommand{\pos}{\mathcal{P}}
\newcommand{\score}{\textup{sc}}
\newcommand{\game}{\mathcal{G}}
\newcommand{\ogame}{\mathbb{G}}
\newcommand{\cA}{\mathcal{A}}
\newcommand{\cB}{\mathcal{B}}
\newcommand{\cE}{\mathcal{E}}
\newcommand{\cH}{\mathcal{H}}
\newcommand{\cI}{\mathcal{I}}
\newcommand{\cV}{\mathcal{V}}
\newcommand{\cX}{\mathcal{X}}
\newcommand{\interval}[2]{\left\llbracket #1,#2 \right\rrbracket}

\newtheorem{theorem}{Theorem}
\newtheorem{lemma}[theorem]{Lemma}

\newtheorem{corollary}[theorem]{Corollary}

\newcommand{\draft}{\textsc{DraftGame}\xspace}
\newcommand{\QBF}{\lang{QBF}\xspace}
\newcommand{\TQBFT}{$3$-\lang{QBF}-$3$\xspace}

\newcommand{\quickset}[1]{\left\lbrace #1 \right\rbrace}
\usepackage{stmaryrd}

\begin{document}
\title{A two-player version of the assignment problem}
%
%

%


\author[1]{Florian Galliot}
\author[2]{Nacim Oijid}
\author[3]{Jonas Sénizergues}

\affil[1]{Aix-Marseille Université, CNRS, I2M, UMR 7373, Marseille, France}

\affil[2]{Umeå University, Umeå, Sweden}

\affil[3]{LaBRI, Université de Bordeaux, Bordeaux, France}

\maketitle              

\begin{abstract}

We introduce the competitive assignment problem, a two-player version of the well-known assignment problem. Given a set of tasks and a set of agents with different efficiencies for different tasks, Alice and Bob take turns picking agents one by one. Once all agents have been picked, Alice and Bob compute the optimal values $s_A$ and $s_B$ for the assignment problem on their respective sets of agents, i.e. they assign their own agents to tasks (with at most one agent per task and at most one task per agent) so as to maximize the sum of the efficiencies. The score of the game is then defined as $s_A-s_B$. Alice aims at maximizing the score, while Bob aims at minimizing it. This problem  can model drafts in sports and card games, or more generally situations where two entities fight for the same resources and then use them to compete against each other. We show that the problem is {\sf PSPACE}-complete, even restricted to agents that have at most two nonzero efficiencies. On the other hand, in the case of agents having at most one nonzero efficiency, the problem lies in {\sf XP} parameterized by the number of tasks, and the optimal score can be computed in linear time when there are only two tasks.

\end{abstract}

\section{Introduction}\strut
\indent\textbf{The assignment problem.} The {\em assignment problem} is a fundamental problem in combinatorial optimization about assigning {\em agents} to {\em tasks}. There must be at most one task per agent and at most one agent per task. An agent $i$ which is qualified for a task $j$ would perform that task with a certain efficiency $e(i,j)$. The goal is to perform as many tasks as possible while maximizing the sum of the respective efficiencies. This problem can be formulated as finding a maximum-weight maximum-cardinality matching in a bipartite graph, with the agents on one side and the tasks on the other. Note that this problem is equivalent to the minimum-weight maximum-cardinality problem in a bipartite graph, up to negating all the weights (and adding some constant if one wants nonnegative weights). As such, the assignment problem can just as well be seen as a cost-minimization problem. 
The first nontrivial algorithm solving the assignment problem may be Easterfield's from 1946, which runs in exponential time \cite{easterfield}. In a highly influential paper from 1955, Kuhn describes what is known as the {\em Hungarian algorithm}, which runs in $O(n^4)$ time where $n$ is the number of vertices of the bipartite graph \cite{kuhn}. Later modifications have brought the running time down to $O(n^3)$ \cite{edmonds-karp,tomizawa}, which remains the best known to this day in all generality.

The assignment problem has applications to many real-world situations, where people and/or machines are appointed different roles on which they work simultaneously and towards the same optimization goal. One can think of a manager giving her employees tasks that best fit their respective skills, a sports coach assembling the best possible lineup from the squad at his disposal, or the chair of a conference program committee assigning papers to reviewers according to their areas of expertise \cite{assignment_reviews}. It also applies, as a time-minimization problem, to any scheduling instance with totally ordered tasks and single-use agents to perform them. Other concrete applications of the assignment problem include assigning patients to hospital beds to maximize the provision of healthcare services \cite{assignment_hospital}, or matchmaking clients to drivers in e-hailing apps to minimize the total distance covered \cite{assignment_examples}.

Variants of the assignment problem have been studied over the time to account for the fact that real-world situations often require assigning tasks to agents in an imperfect environment. In particular, optimal strategies have been found for the {\em online} version of the problem, where the graph is not known at the start but instead the tasks (or the agents, in symmetrical fashion) arrive one by one in random order and must be assigned instantly without knowing what is coming next \cite{fahrbach,karp,kesselheim,korula}. Another variant, where conflicts must be avoided between some agent-task pairs, has been shown to be {\sf NP}-hard but tractable on a large class of realistic instances \cite{conflicts}. Finally, if an attacker waits until the assignment is done to attack some of the agents and reduce their performance, then we again get an {\sf NP}-hard problem \cite{adversarial}. That last variant can be seen as a single-round game between the assigner and the attacker and, as far as we know, constitutes the only prior example of a two-player version of the assignment problem.

\medskip

\indent\textbf{Our contribution.} We introduce the {\em competitive assignment problem}, which models situations where two entities need the same tasks to be performed for their own and are competing for the same agents. This can be seen as a multiple-round two-player game between, say, Alice and Bob. The players take turns picking agents one by one, with Alice going first. Once all the agents have been picked, Alice and Bob each solve the assignment problem on their own team of agents, with optimal values $s_A$ and $s_B$ respectively. The {\em score} of the game is defined as $s_A-s_B$. Alice aims at maximizing the score, while Bob aims at minimizing it. This is effectively a {\em draft} process, similar to what is used in many sports leagues (such as the NBA) to allocate players to teams, albeit involving two teams only. We thus propose, to our knowledge, the first mathematical model and study of optimal strategies for drafts, apart from the specific case of multiplayer online battle arena games \cite{MOBA2,MOBA1} in which the optimized function is very different.

Our definition of the game, and of the score in particular, is suited to the case where the two teams will compete against each other after the selection is done: as such, the goal is not to have the best team in absolute terms, but instead to have the biggest possible edge over the opponent. This applies to fantasy sports as well as various card games which involve a draft phase. Additionally, real-world applications include competitive environments where there is a need to identify which resources should be secured urgently before rivals snatch them away. An example would be two companies in the same industry and with the same customer market, trying to prioritize as to which jobs to fill and which partnerships to strike.

The competitive assignment problem closely relates to what we may call the {\em fair 2-assignment problem}, where one must build two agent-disjoint assignments, each assigning all the tasks, such that the two total weights are as equal as possible. Such problems may arise, for instance, when a unit gets divided into two departments so that its human and material resources must be split between the two. It can easily be seen that the fair 2-assignment problem generalizes the well-known {\em balanced 2-partition} problem and, as such, is {\sf NP}-complete \cite{gareyjohnson}. The competitive assignment problem provides one way to approach it, by getting two imaginary players to fight for the agents, with the intuition that the final split should often be quite equal. By designing strategies for the player at a disadvantage (this will be Bob), we at least get bounds on the optimal balance. More generally, we believe that scoring games are natural and useful companions for optimization problems about partitioning a set into two parts which must be balanced in some way. For instance, consider the problem of partitioning the vertices of a graph into two equal parts while minimizing the difference between the number of edges they induce. A bound for this nontrivial problem is given by the optimal score of its associated two-player scoring game (i.e. the players draft the vertices and the score is the number of edges induced by Alice's vertices minus those induced by Bob's), which can be computed in polynomial time \cite{incidence}.

\medskip

\indent\textbf{Overview of the results.} A few general properties of the competitive assignment problem are listed in Section \ref{section2}. Some of them follow from its membership in two classes of combinatorial games: Milnor's universe, and Maker-Breaker positional games. 
We also establish some basic strategic principles, which narrow down the candidate moves in some situations. The next two sections investigate the algorithmic complexity of the associated decision problem, which asks whether the optimal score is greater or equal to some integer $s$. In Section \ref{section3}, we show that this problem is {\sf PSPACE}-complete, even restricted to agents that have at most two nonzero efficiencies. Section \ref{section4} thus addresses agents that have at most one nonzero efficiency. We show that this case is in {\sf XP} parameterized by the number of tasks, and we get a linear-time algorithm when there are only two tasks. Section \ref{section5} concludes the paper and lists some perspectives for future research.

\section{Framework and first general results}\label{section2}

\subsection{The draft game}

From here on, the competitive assignment problem will be referred to as the {\em draft game} for short. An instance $\game$ of the draft game is a set of $n$ {\em agents}, each of which is a $t$-dimensional vector of real numbers, where each dimension represents a {\em task} (in particular, $t$ is the number of tasks). Alice and Bob take turns picking agents one by one, with Alice going first, until the set of agents has been partitioned into the set $\cA$ of Alice's agents and the set $\cB$ of Bob's agents. Denoting by $X^{(i)}$ the $i$-th coordinate of an agent $X$, we define the {\em score} of the game as:
$$ \max_{\substack{ a : [t] \to \cA \\ \text{$a$ injective} }} \sum_{i=1}^{t} a(i)^{(i)} - \max_{\substack{ b : [t] \to \cB \\ \text{$b$ injective} }} \sum_{i=1}^{t} b(i)^{(i)}, $$
which can be computed in polynomial time given $\cA$ and $\cB$ by solving two separate instances of the (classical) assignment problem. The {\em optimal score} of $\game$, denoted by $\score(\game)$, is the score obtained when Alice and Bob apply optimal strategies, given that Alice wants to maximize the score while Bob wants to minimize it. 

A {\em position} of the draft game is a triple $\pos=(\game,\cA,\cB)$, where $\game$ is an instance of the draft game and $\cA,\cB$ are disjoint subsets of $\game$ which we interpret as the sets of agents that have already been picked by Alice and Bob respectively. When considering instances $\game$, we will assume that Alice plays first, whereas when considering positions $\pos$, we will specify who is next to play. Note that an instance $\game$ of the draft game could be seen as a starting position $(\game,\varnothing,\varnothing)$.

Since the game is largely unchanged under adding the same constant to every efficiency, we will assume throughout the paper that the agents are vectors of nonnegative numbers.

Note that, in our model, any agent can perform any task (with efficiency zero at worst), unassigned tasks are allowed and count for zero, and agents are assigned tasks at the very end rather than at the moment they are picked. These choices will be discussed in the conclusion. Also note that, since two agents could have the same efficiencies for all tasks, an instance $\game$ should technically be defined as a multiset, but we opt not to in order to alleviate notations.

As we will be particularly interested in agents that have few nonzero efficiencies, we introduce the following definitions. An {\em OTP} (which stands for {\em one-trick poney}) is an agent that has at most one nonzero efficiency. Similarly, a {\em TTP} is an agent that has at most two nonzero efficiencies.

Following standard practice for optimization problems, we define an associated decision problem \draft which, given an instance $\game$ of the draft game and an integer $s$, asks whether $\score(\game) \geq s$. In other words, ``does Alice have a strategy ensuring a score of at least $s$ against any strategy of Bob?''

\subsection{Membership in Milnor's universe}

Milnor's universe~\cite{milnor1953} comprises of scoring games that are both {\em dicotic}, i.e., from any given position of the game, Alice has a legal move as the next player if and only if Bob has a legal move as the next player; and {\em nonzugzwang}, i.e., from any given position of the game, the optimal score with Alice as the next player is greater or equal to the optimal score with Bob as the next player.

\begin{lemma} \label{milnor}
        The draft game belongs to Milnor's universe.
\end{lemma}

\begin{proof}
Any available agent can be picked by any player, therefore the game is dicotic. We now prove that it is nonzugzwang.

    Let $\pos$ be a position of the draft game. Suppose that, going second on $\pos$, Alice has a strategy $\strat$ which ensures a score greater or equal to $s$. We show that she can also ensure a score greater or equal to $s$ when going first. The proof is a classical {\em strategy-stealing} argument. Going first, Alice claims an arbitrary agent $X_0$. Since Bob is now next to play, Alice can pretend like she did not claim $X_0$ and play according to $\strat$ from here on, as follows:
\begin{itemize}[noitemsep,nolistsep]
    \item If $\strat$ dictates to claim some agent $X \neq X_0$, then Alice claims $X$.
    \item If $\strat$ dictates to claim $X_0$, then Alice claims another arbitrary agent $X$ instead, and continues to follow $\strat$ as if $X$ was the arbitrary agent that she had claimed on her first turn.
\end{itemize}

Note that following this strategy is always possible, since every one of Bob's moves during the game could also have been made had Alice not claimed $X_0$. At the end of the game, the set of agents claimed by Alice contains all the agents she would have claimed going second, and symmetrically the set of agents claimed by Bob is included in the set of agents he would have claimed going first. Thus, since the assignment problem is monotone under taking subgraphs, the score will be greater or equal to $s$.
 \end{proof}

\begin{corollary}\label{corollary:score-positif}
    For any instance $\game$ of the draft game, we have $\score(\game) \geq 0$.
\end{corollary}

\begin{proof}
From an instance (i.e. a starting position) $\game$ of the draft game, any strategy for a player can be used by the other player if we change the order of play. Therefore, since the optimal score when Alice is the first player is $\score(\game)$ by definition, the optimal score when Alice is the second player is $-\score(\game)$. As the draft game is nonzugzwang by Lemma \ref{milnor}, it follows that $ \score(\game) \ge -\score(\game)$ hence $\score(\game) \ge 0$.
 \end{proof}

Milnor's universe allows for defining powerful tools to study scoring games: we refer the interested reader to Stewart's thesis~\cite{stewart2012} and Duchêne et al.~\cite{Duchene2024}. One of them is the notion of {\em mean}, introduced by Hanner in 1959~\cite{hanner1959}, which predates the current formalism and standard notations in scoring game theory. This notion relies on the definition of a {\em sum} of games.

The sum of two games $\ogame_1$ and $\ogame_2$, denoted by $\ogame_1 + \ogame_2$, is the game in which Alice and Bob have the choice to move either in $\ogame_1$ or in $\ogame_2$ and the score of the game is the sum of the scores in $\ogame_1$ and in $\ogame_2$. Milnor's universe is an abelian group with regard to this sum operator~\cite{milnor1953}.
We denote by $n \ogame$ the game obtained by summing $n$ copies of the game $\ogame$.

Intuitively, the mean of a game $\ogame$ is the average score $m(\ogame)$ that could be expected if the players have to play on a large number of copies of $\ogame$. Formally:

$$ m(\ogame) = \underset{n\to \infty}{\lim} \dfrac{\score(n\ogame)}{n},$$
which is well defined since this limit always exists in Milnor's universe~\cite{hanner1959}.

In our context of the draft game, playing on a sum of instances (or a sum of positions, more generally) means considering agents that are all drafted together but work on separate sets of tasks, where the final score is the sum of the scores obtained on each set of tasks. Even though the mean can be nonzero for some positions of the draft game, the next result shows that it is always zero for starting positions, indicating that the draft game is somehow balanced.

\begin{lemma}\label{lemma:mean-zero}
    For any instance $\game$ of the draft game, we have $m(\game) = 0$.
\end{lemma}

\begin{proof}
    Let $u_n = \frac{\score(n\game)}{n}$. We show that $u_{2n}=0$ for all $n \geq 1$. Since $(u_n)_n$ converges by~\cite[Theorem 4]{hanner1959}, this is enough to prove the lemma.

    We know that $\score(2n\game) \geq 0$ by Corollary \ref{corollary:score-positif}. Therefore, it suffices to show that Bob has a strategy ensuring a score of 0 on $2n\game$. Let $\game_1,\ldots,\game_{2n}$ be identical copies of $\game$, which we pair as $(\game_1, \game_2), (\game_3,\game_4), \dots, (\game_{2n-1}, \game_{2n})$. For all $k \in\interval{1}{2n}$, define $k'$ as $k' = k+1$ if $k$ is odd or $k' = k-1$ if $k$ is even. Write $\game_k=\{X_1^k,\ldots,X_m^k\}$, so that the agents $X_i^1,\ldots,X_i^{2n}$ are identical for all $i \in\interval{1}{m}$. Whenever Alice picks an agent $X_i^k$, Bob picks the agent $X_i^{k'}$. Once all the agents are picked, for all $k \in \interval{1}{2n}$, the agents picked by Alice in $\game_k$ are exactly the agents picked by Bob in $\game_{k'}$. Therefore, the score is 0 on each pair, so it is 0 across the entire sum. 
\end{proof}

\subsection{The draft game as a Maker-Breaker positional game}

Maker-Breaker games were introduced by Erd\H{o}s and Chv\'atal in 1978 as natural two-player versions of some central problems in extremal combinatorics \cite{chvatalerdos}. The board is a hypergraph $\cH=(\cV,\cE)$, of which Maker and Breaker take turns claiming vertices, with Maker going first. Maker wins if and only if she manages to claim a whole hyperedge. Given an instance $\cI$ of \draft i.e. a set $\game$ of $n$ agents and a threshold score $s$, we can define a hypergraph $\cH_s=(\cV,\cE_s)$ as follows. The vertices are simply the agents. The hyperedges are all the sets $\cX \subseteq \cV$ of $\lceil\frac{n}{2}\rceil$ agents such that $\textup{opt}(\cX)-\textup{opt}(\cV \setminus \cX) \geq s$, where $\textup{opt}(\cX)$ returns the optimal value for the classical assignment problem restricted on the set of agents $\cX$. By definition of $\cH_s$, Maker has a winning strategy for the Maker-Breaker game on $\cH_s$ if and only if $\cI$ is a yes-instance.

Since the size of $\cH_s$ is exponential in the size of $\cI$, this remark is of no help with regard to the algorithmic complexity of the draft game. General Maker-Breaker games have long been known to be {\sf PSPACE}-complete anyway \cite{schaefer}. However, when it comes to strategic principles, the draft game inherits some general properties of Maker-Breaker games. One such example is actually its nonzugzwang nature, already stated in Lemma \ref{milnor}. Another property of Maker-Breaker games is that, if two vertices $u$ and $v$ are such that every hyperedge containing $v$ also contains $u$, then $u$ is at least as good a move as $v$. The next lemma formulates this statement in the context of the draft game.

\begin{lemma} \label{lemma dominating agents}
    Let $\pos$ be a position of the draft game, and let $X = (x_1, \dots, x_t)$, $Y = (y_1, \dots, y_t)$ be two free agents such that $ x_i \geq y_i$ for all $i \in \interval{1}{t}$. Let $s$ be an integer. If Alice is the next player and can ensure a score of at least $s$ by picking $Y$, then she also can by picking $X$. Similarly, if Bob is the next player and can ensure a score of at most $s$ by picking $Y$, then he also can by picking $X$. Should this happen, we say that $X$ {\em dominates} $Y$ or that $Y$ {\em is dominated by} $X$. 
\end{lemma}

\begin{proof}
This is a classical proof of ``dominated move'' in scoring game theory. By symmetry of the argument, assume that Alice is the next player. Let $\strat$ be a strategy for Alice that starts by picking $Y$ and ensures a score of at least $s$. Consider the following strategy for Alice:

\begin{itemize}[noitemsep,nolistsep]
    \item Alice picks $X$ instead of $Y$, and then starts applying $\strat$ as if she had claimed $Y$.
    \item If, at some point before Bob has picked $Y$, $\strat$ dictates to pick $X$, then Alice picks $Y$ instead and she continues to apply $\strat$.
    \item If, at some point before $\strat$ has dictated to pick $X$, Bob picks $Y$, then Alice pretends that Bob has picked $X$ instead and she continues to apply $\strat$.
\end{itemize}

Following this strategy, once all the agents have been picked, the respective sets of agents picked by Alice and Bob are the same as if $\strat$ had been executed normally, except that $X$ and $Y$ may have been exchanged. If Alice has picked both $X$ and $Y$, then these sets are actually the same, so the score is the same. If $X$ and $Y$ have been exchanged, then the score is at least as high, since the classical assignment problem is monotone under replacing an agent by one that is better at all tasks. 
Thus, picking $X$ instead of $Y$ was optimal.
 \end{proof}

\subsection{Other properties of the draft game}

General bounds for the optimal score can be computed in linear time.
\begin{lemma}\label{lemma:boundscore}
    Every instance $\game$ of the draft game satisfies $0 \leq \score(\game) \leq \underset{X \in \game}{\max} \,\lVert X \rVert_{\infty}$.
\end{lemma}

\begin{proof}
By Corollary \ref{corollary:score-positif}, we already know that $\score(\game) \geq 0$. We now show that Bob has a strategy ensuring a score of at most $\underset{X \in \game}{\max} \,\lVert X \rVert_{\infty}$. Let $X_0 = (x_1, \dots, x_t)$ be the first agent picked by Alice. Define an instance $\game' = \game \setminus \{ X_0\}$ where $X_0$ does not exist, and let $s_0 = sc(\game')$. After Alice has picked $X_0$, Bob applies a strategy that is optimal for $\game'$. By doing so, he ensures that the score if Alice could not use her agent $X_0$ would be at most $-s_0$. Since she actually can use $X_0$, the score is at most $ - s_0 + \lVert X_0 \rVert_{\infty} \leq \lVert X_0 \rVert_{\infty}$.  Let $\pos$ be a position and $\cX$ be the set of free agents in $\pos$.
If there is a free agent that has the highest efficiency for both tasks, then that agent dominates all the agents, so it is optimal to pick it by Lemma \ref{dominating agent}. Therefore, suppose that this is not the case. Consider two agents $X_1=(x_1,y_1')$, $Y_1=(x_1',y_1) \in \cX$ that have the highest efficiency respectively in the first task and in the second task.

Let $U=(x,y) \in \cX \setminus \quickset{X_1,Y_1}$. If $U$ is dominated by $X_1$ or $Y_1$, then by Lemma~\ref{dominating agent} there is nothing prove. Therefore, assume that $U$ is dominated by neither $X_1$ nor $Y_1$, i.e. $x_1 \geq x \geq x_1'$ and $y_1 \geq y \geq y_1'$. Suppose for a contradiction that $U$ is an optimal pick, but that $X_1$ and $Y_1$ are not.

We may suppose that $U$ is not dominated by any other agent, as any agent dominating $U$ would also be optimal. We write $x_2$ (resp. $y_2$) the best efficiency for the first task (resp. second task) among all agents except $X_1$, $Y_1$, and $U$.
\end{proof}

The next lemma formalizes the idea that it is always optimal to pick an agent which impacts the score so dramatically that it cannot be compensated for the entire rest of the game.

\begin{lemma}\label{dominating agent}
    Let $\pos$ be a position of the draft game. Let $X_1, \dots, X_m$ be the free agents in $\pos$, where $X_i = (x^i_1, \dots, x^i_t)$. For all $i \in \interval{1}{t}$, let $\alpha_i$ and $\beta_i$ denote the highest efficiencies for the $i$-th task among the agents already picked by Alice and Bob respectively. 
    Suppose that there exists $(i,j) \in \interval{1}{t} \times \interval{1}{m}$ such that $\max (x^j_i - \alpha_i, x^j_i - \beta_i ) \ge 2\underset{\substack{1 \leq \ell \leq m \\ \ell \ne j}}{\sum} \underset{1 \le k \le t}{\max} (x^{\ell}_k)$. Then, picking $X_j$ is optimal for whoever is the next player, and we say that the agent $X_j$ is {\em dominating} in the position $\pos$.
\end{lemma}

\begin{proof}
    By symmetry of the argument, assume that Alice is the next player.
    Let $s_0$ be the provisional optimal score, i.e. the one that would be obtained if Alice and Bob could only use the agents that they have already picked in the position $\pos$. Let $s_1$ (resp. $s_2$) be the optimal score in the scenario where Alice picks $X_j$ (resp. does not pick $X_j$) on her next move.

    Suppose that Alice picks $X_j$ on her next move. Then, she could improve her provisional assignment by $\max(0,x_i^j-\alpha_i)$, since she could choose to assign $X_j$ to the $i$-th task at the end. Meanwhile, Bob will improve his by at most $\underset{1 \le k \le t}{\max} (x^{\ell}_k)$ for each free agent $X_{\ell}$ that he picks. All in all, we get:
    $$ s_1 \geq s_0 + \max(0,x_i^j-\alpha_i) - \underset{\substack{1 \leq \ell \leq m \\ \ell \ne j}}{\sum} \underset{1 \le k \le t}{\max} (x^{\ell}_k). $$
    
    Suppose now that Alice does not pick $X_j$ on her next move. Then, Bob could pick $X_j$ himself, which yields in analogous fashion:
    $$ s_2 \leq s_0 + \underset{\substack{1 \leq \ell \leq m \\ \ell \ne j}}{\sum} \underset{1 \le k \le t}{\max} (x^{\ell}_k) - \max(0,x_i^j-\beta_i). $$

    We now use the assumption that $\max (x^j_i - \alpha_i, x^j_i - \beta_i ) \ge 2\underset{\substack{1 \leq \ell \leq m \\ \ell \ne j}}{\sum} \underset{1 \le k \le t}{\max} (x^{\ell}_k)$. If $x^j_i - \alpha_i \ge 2\underset{\substack{1 \leq \ell \leq m \\ \ell \ne j}}{\sum} \underset{1 \le k \le t}{\max} (x^{\ell}_k)$, then:
    $$ s_1 \geq s_0 + \max(0,x_i^j-\alpha_i) - \underset{\substack{1 \leq \ell \leq m \\ \ell \ne j}}{\sum} \underset{1 \le k \le t}{\max} (x^{\ell}_k) \geq s_0 + \underset{\substack{1 \leq \ell \leq m \\ \ell \ne j}}{\sum} \underset{1 \le k \le t}{\max} (x^{\ell}_k) \geq s_2. $$
    If $x^j_i - \beta_i \ge 2\underset{\substack{1 \leq \ell \leq m \\ \ell \ne j}}{\sum} \underset{1 \le k \le t}{\max} (x^{\ell}_k)$, then:
    $$ s_2 \leq s_0 + \underset{\substack{1 \leq \ell \leq m \\ \ell \ne j}}{\sum} \underset{1 \le k \le t}{\max} (x^{\ell}_k) - \max(0,x_i^j-\beta_i) \leq s_0 - \underset{\substack{1 \leq \ell \leq m \\ \ell \ne j}}{\sum} \underset{1 \le k \le t}{\max} (x^{\ell}_k) \leq s_1. $$

    In both cases, we get $s_1 \geq s_2$, which concludes the proof.
 \end{proof}

\begin{lemma} \label{lemma: paired forcing move} 
Let $\pos$ be a position of the draft game, and let $X,Y$ be two free agents in $\pos$.
Suppose that both players have optimal strategies from the position $\pos$ such that, as soon as a player picks an agent in $\{X,Y\}$, the opponent immediately picks the other agent in $\{X,Y\}$. Then, there exist optimal strategies for both players from the position $\pos$ where they start by picking $X$ and $Y$ (in some order), and we say that $\{X,Y\}$ is a {\em dominating pair} in the position $\pos$.
\end{lemma}

\begin{proof}
    By symmetry of the argument, assume that Alice is the next player. Let $\strat$ be an optimal strategy for Alice from the position $\pos$, and let $Z$ be the agent that $\strat$ dictates to pick first.
 Assume that $Z \not\in \{X,Y\}$, otherwise we are already done.

 If Alice picks $Z$, then Bob could pick $Y$, which would prompt Alice to pick $X$ by assumption.  Therefore, the optimal score if Alice picks $Z$ is at most the optimal score playing second on $\pos'$, defined as the position obtained from $\pos$ by giving $X$ and $Z$ to Alice and $Y$ to Bob.

 If Alice picks $X$ instead, then Bob will pick $Y$ by assumption. After that, Alice could pick $Z$, and we would reach $\pos'$ again. Therefore, the optimal score if Alice picks $X$ is at least the optimal score playing second on $\pos'$.

 By comparing the two options, we see that $X$ is at least as good a pick as $Z$ for Alice.
\end{proof}

\begin{lemma}\label{lemma:first-move-two-tasks}
Let $\pos$ be a position of the draft game with only two tasks. Then, there exists a free agent that has the highest efficiency among all free agents for at least one of the tasks and that is an optimal pick for the next player.
\end{lemma}

\begin{proof}
 Observe that only the top two efficiencies for each task can be used in the final assignment for each player: either the two top efficiencies belong to two different agents and thus the player has no reason to use any of the other agents, or the top two efficiencies belong to the same agent, and then the player may want to use the second best efficiency for either the first or the second task. Thus, instead of considering a full position, we will only consider the following reduced position:

Given a position $\pos=(\game,\cA,\cB)$ on a $\game$ on two tasks, and a current player $p$, we build a reduced position 
$(\cX,x_{A,0},x_{A,1},y_{A,0},y_{A,1},x_{B,0},x_{B,1},y_{B,0},y_{B,1},p)$:
\begin{itemize}
\item $\cX = \game \setminus \left( \cA \cup \cB \right)$ is the set of free agents,
\item For Alice, $x_{A,0} \geq x_{A,1}, y_{A,0} \geq y_{A,1}$ are such that:
	\begin{itemize}
	\item If Alice's top efficiencies for both tasks belong to the same agent they are respectively the first and second best efficiencies for the first task and the first and second best efficiencies for the second task among Alice's agents.
	\item If the top efficiencies belong to different agents picked by Alice, $x_{A,0} = x_{A,1}$ is the best efficiency for the first task among Alice's agents, and $y_{A,0}=y_{A,1}$ is the best efficiency for the second task among Alice's agents.
	\end{itemize}
\item For Bob, $x_{B,0} \geq x_{B,1}, y_{B,0} \geq y_{B,1}$ are defined as for Alice but with Bob's agents.
\item $p$ is the player who is next to play.
\end{itemize}

When no agents are remaining to be picked, these two efficiencies are sufficient to compute the final score, which is $\underset{i \in \interval{0}{1}}{\max}(x_{A,i}+y_{A,i+1}) - \underset{i \in \interval{0}{1}}{\max}(x_{B,i}+y_{B,i+1})$. Moreover, after a player, say Alice, picks an agent $(x,y)$, we can update the reduced position without referring to the full position. We have to remove the picked agent from the set of not picked agents, Bob's efficiencies are kept, and:

\begin{itemize}
\item If $x>x_{A,0}$ and $y>y_{A,0}$ then we can update $x_{A,0}$ to $x$, $x_{A,1}$ to the old value of $x_{A,0}$, $y_{A,0}$ to $y$ and $y_{A,1}$ to the old value of $y_{A,0}$. Note that under optimal play, this situation cannot occur after the first move as it can only happen if $(x,y)$ dominates all the agents already picked by Alice.
\item Else,
	\begin{itemize}
	\item If $x>x_{A,0}$, we set $x_{A,0}$ to $x$, and $x_{A,1}$ to the old value of $x_{A,0}$,
	\item If $x_{A,0} \geq x > x_{A_1}$, we keep $x_{A,0}$ and set $x_{A,1}$ to $x$,
	\item If $x_{A,1} \geq x$, we keep both $x_{A,0}$ and $x_{A,1}$ efficiencies,
	\item The same is done for the efficiencies for the second task.
	\end{itemize}
\end{itemize}

The same can be done when Bob picks an agent.

Then, we can consider those reduced positions to prove our lemma.

Consider a position $\pos$ and its associated reduced position with $\cX$ the set of free agents, and $x_{A,0} \geq x_{A,1}, y_{A,0} \geq y_{A,1}, x_{B,0} \geq x_{B,1}, y_{B,0} \geq y_{B,1}$. Without loss of generality, we will suppose that Alice is the next to pick, the result for Bob will follow by symmetry.

If there is a free agent that has the highest efficiency for both tasks, then that agent dominates all the agents, so it is optimal to pick it by Lemma \ref{dominating agent}. Therefore, suppose that this is not the case. Consider two agents $X_1=(x_1,y_1')$, $Y_1=(x_1',y_1) \in \cX$ that have the highest efficiency respectively in the first task and in the second task.

If for any other agent $U \in \cX$, $U$ is either dominated by $X_1$ or $Y_1$ then it is always better to pick $X_1$ or $Y_1$ and there is nothing to prove.

Let $U=(x,y) \in \cX \setminus \quickset{X_1,Y_1}$. If $U$ is dominated by $X_1$ or $Y_1$, then by Lemma~\ref{dominating agent} there is nothing prove. Therefore, assume that $U$ is dominated by neither $X_1$ nor $Y_1$, i.e. $x_1 \geq x \geq x_1'$ and $y_1 \geq y \geq y_1'$. Suppose for a contradiction that $U$ is an optimal pick, but that $X_1$ and $Y_1$ are not.

We may suppose that $U$ is not dominated by any other agent in $\cX$, as any agent dominating $U$ would also be optimal. We write $x_2$ (resp. $y_2$) the best efficiency for the first task (resp. second task) among all agents except $X_1$, $Y_1$, and $U$.

If Alice picks $U$ first:

By picking $X_1$, Bob guarantees a score of at most (addition in indices are considered modulo~$2$):
\begin{equation}
\max_{i \in \interval{0}{1}}(\max(x,x_2,x_{A,i})+\max(y_1,y_{A,i+1})) - \max_{i \in \interval{0}{1}}(\max(x_1,x_{B,i})+\max(y_2,y_{B,i+1}))
\label{score_u_x1}
\end{equation}

By picking $Y_1$, Bob guarantees a score of at most:
\begin{equation}
\max_{i \in \interval{0}{1}}(\max(x_1,x_{A,i})+\max(y,y_2,y_{A,i+1})) - \max_{i \in \interval{0}{1}}(\max(x_2,x_{B,i})+\max(y_1,y_{B,i+1}))
\label{score_u_y1}
\end{equation}

If Alice picks $X_1$ and Bob answers by picking $Y_1$, by picking the best one among $U$ and the best agent with efficiency $y_2$ for the second task Alice guarantees a score of at least: 
\begin{equation}
\max_{i \in \interval{0}{1}}(\max(x_1,x_{A,i})+\max(y,y_2,y_{A,i+1})) - \max_{i \in \interval{0}{1}}(\max(x,x_2,x_{B,i})+\max(y_1,y_{B,i+1}))
\label{score_x1_y1}
\end{equation}

If Alice picks $Y_1$ and Bob answers by picking $X_1$, by picking the best one among $U$ and the best agent with efficiency $x_2$ for the second task Alice guarantees a score of at least: 
\begin{equation}
\max_{i \in \interval{0}{1}}(\max(x,x_2,x_{A,i})+\max(y_1,y_{A,i+1})) - \max_{i \in \interval{0}{1}}(\max(x_1,x_{B,i})+\max(y,y_2,y_{B,i+1}))
\label{score_y1_x1}
\end{equation}

If Alice picks $X_1$ or $Y_1$ and Bob does not pick the other one, by grabbing the other one Alice guarantees a score of at least: 
\begin{equation}
\max_{i \in \interval{0}{1}}(\max(x_1,x_{A,i})+\max(y_1,y_{A,i+1})) - \max_{i \in \interval{0}{1}}(\max(x,x_2,x_{B,i})+\max(y,y_2,y_{B,i+1}))
\label{score_both_good}
\end{equation}

Observe that the value of expression~\ref{score_both_good} is always greater, \textit{i.e.} better for Alice, than both the expressions~\ref{score_x1_y1} and~\ref{score_y1_x1}, thus we may suppose without loss of generality that Bob answers by the other when Alice picks either $X_1$ or $Y_1$.

Moreover, when $x\leq x_2$, expression~\ref{score_x1_y1} become the same as expression~\ref{score_u_y1}, thus by picking $X_1$ first Alice can always guarantee a score greater than or equal to a value Bob can guarantee the score to be lower than or equal to if she picks $U$ first, contradiction with none of $X_1$ and $Y_1$ being optimal. 

Symmetrically, when $y\leq y_2$, there is also a contradiction with none of $X_1$ and $Y_1$ being optimal.

Thus, we can now suppose that $x>x_2$ and $y>y_2$. In this case, we go back to the scenarios where Alice picks either $X_1$ or $Y_1$ first. As we said earlier, we can suppose Bob answers by picking the other one. Then, Alice picks $U$.

If she picks $X_1$, this guarantees a score of at least:
\begin{equation}
\max_{i \in \interval{0}{1}}(\max(x_1,x_{A,i})+\max(y,y_{A,i+1})) - \max_{i \in \interval{0}{1}}(\max(x_2,x_{B,i})+\max(y_1,y_{B,i+1}))
\label{score_x1_y1_u}
\end{equation}

Then simplifying expression~\ref{score_u_y1} with the fact that $y>y_2$, it becomes the same as expression~\ref{score_x1_y1_u}. Thus, by picking $X_1$ first Alice can always guarantee a score greater than or equal to a value Bob can guarantee the score to be lower than or equal to if she picks $U$ first, contradiction with none of $X_1$ and $Y_1$ being optimal.

Thus, by contradiction, at least one of $X_1$ or $Y_1$ must be an optimal play for Alice. Symmetrically, it is also true when Bob is next to play.
 \end{proof}

Even though Lemma \ref{lemma:first-move-two-tasks} ensures that it is always optimal to pick an agent which maximizes one of the two tasks, that agent will not necessarily be assigned to that task in the end. For instance, consider $\game=\{X,Y,Z\}$ where $X = (4,7)$, $Y = (5,5)$ and $Z =(0,4)$. It is straightforward that $\score(\game)=3$, and that the only optimal sequence of moves is the following: Alice picks $X$ (which maximizes the second task), Bob picks $Y$, Alice picks $Z$. Alice then assigns $X$ to the first task and $Z$ to the second task, to get a score of $4+4-5=3$.

Finally, we note that Lemma \ref{lemma:first-move-two-tasks} does not generalize to instances with three tasks. A counterexample is provided by $\game=\{X_1,X_2,X_3,X_4,X_5,X_6\}$, where: $X_1=(5,0,0)$, $X_2=(0,5,0)$, $X_3=(0,0,5)$, $X_4=(4,4,4)$, $X_5=(0,3,3)$ and $X_6=(3,0,0)$. It can be checked that $\score(\game)=2$, with $X_4$ being the only optimal fist pick for Alice despite $X_4$ maximizing none of the tasks.

\section{\PSPACE-completeness for TTPs}\label{section3}

\begin{theorem}\label{theo_pspace}
    \draft~is \PSPACE-complete, even restricted to agents with at most two nonzero efficiencies.
\end{theorem}

\begin{proof}
Determining whether the optimal score of a given instance of the draft game is greater or equal to some integer $s$ can be decided by a polynomial number of choices (each agent is only picked once) of alternating choices between Alice and Bob, followed by solving the classical assignment problem for both players.
In other words, this decision problem can be expressed using existential quantifiers for Alice's choices and universal quantifiers for Bob's choices, followed by a polynomial-time computation.
Therefore, \draft belongs to the class \AP~of problems decidable in polynomial time by an alternating Turing machine. Since $\AP = \PSPACE$~\cite{chandra1981alternation}, membership in \PSPACE~ensues.

The proof of hardness is a reduction from \TQBFT, which is a special case of \QBF. In the \QBF~game, given a formula $\varphi$ with variables $x_1,y_1,\ldots,x_n,y_n$, Satisfier and Falsifier take turns choosing valuations for $x_1,y_1,\ldots,x_n,y_n$ in that order, with Satisfier going first. Satisfier wins if and only if the formula ends up satisfied. The subproblem \TQBFT, which addresses formulas in conjunctive normal form where each clause has size at most $3$ and each variable appears exactly three times, has been proved to be \PSPACE-complete by Oijid~\cite{Oijid2025}.

Up to removing variables appearing only positively or only negatively (whose valuation can easily be determined depending on their quantifier), we assume that each variable appears at least once and at most twice positively and negatively. Moreover, up to exchanging $x_i$ and $\neg x_i$, which is possible without altering the truth value of the formula, we assume that each variable appears exactly twice positively and once negatively. Although not mandatory for the proof to work, these assumptions will make the reduction easier to describe.

When an agent has a nonzero efficiency in one task only, we can always assume that this agent is assigned to that specific task. Moreover, we can also assume that an agent is never assigned to a task in which an agent with higher efficiency is already assigned. If the two arguments can be applied at the same time, then the agent will just not contribute to the total score of the player. Hence, most of the time in the proof, we will not specify the task that an agent is assigned to,, as the optimal choice will be obvious.

Let $\psi = \exists x_{1} \forall y_1 \dots \exists x_n \forall y_{n} \varphi$ be a quantified formula where $\varphi = C_1 \wedge \dots \wedge C_m$ is quantifier-free. 

\subsection*{Overview of the construction:} 

The main idea of the proof is that Alice will mimic Satisfier's choices and will score one more point with each satisfied clause. The construction $\game$ uses three types of gadgets (each gadget being a set of agents), pictured in Figure \ref{fig:table-pspace-complete}:

\begin{itemize}[nolistsep,noitemsep]
    \item A unique {\em setup gadget} $\mathcal{U}$, consisting in $2m+2$ agents and $m+2$ tasks, ensuring that 
    Bob scores $1$ per clause task, which will be used to determine whether Alice satisfies the formula or not.
    \item For all $i \in \interval{1}{n}$, {\em Alice's choice gadget} $\mathcal{A}_i$, each containing seven agents and two new tasks.
    \item For all $i \in \interval{1}{n}$, {\em Bob's choice gadget} $\mathcal{B}_i$, each containing nine agents and four new tasks.
\end{itemize}
When defining agents, we will only specify their nonzero efficiencies. The agents' efficiencies in the different tasks are chosen using Lemma~\ref{dominating agent} and \ref{lemma: paired forcing move}, so that optimal play respects the order $\mathcal{U}, \mathcal{A}_1, \mathcal{B}_1, \mathcal{A}_2, \mathcal{B}_2, \dots, \mathcal{A}_n, \mathcal{B}_n$ on the gadgets, and all agents of a gadget are picked before play starts in the next gadget.

\subsection*{Setup gadget}

We first introduce two tasks $A, B$, and two agents $A_1,B_1$ such that $A_1$ has efficiency $\alpha$ in $A$ and $B_1$ has efficiency $\beta$ in $B$. We then introduce $m$ tasks $S_1, \dots, S_m$ and $2m$ agents $\Gamma_1, \Gamma_1', \dots, \Gamma_m, \Gamma_m'$, such that Agent $\Gamma_i$ has efficiency $\gamma_i$ in task $A$, and  Agent $\Gamma_i'$ has efficiency $\gamma_i'$ in task $B$ and efficiency $1$ in task $S_i$.

\begin{figure}[t!]
    \centering
    \subfloat[][The setup gadget.]{
    \centering
    \begin{tabular}{c||c|c|c|c|c|c}
                    & A & B & $S_1$ & $S_2$ & $\dots$ & $S_m$ \\ \hline
$A_1$               & $\alpha$ & 0 & 0 & 0 & \dots & 0 \\ \hline
$B_1$               & 0 & $\beta$ & 0 & 0 & \dots & 0 \\ \hline    
$\Gamma_1$          & $\gamma_1$ & 0 & 0 & 0 & \dots & 0 \\ \hline
$\Gamma_1'$         & 0 & $\gamma_1'$ & 1 & 0 & \dots & 0 \\ \hline
$\Gamma_2$          & $\gamma_2$ & 0 & 0 & 0 & \dots & 0 \\ \hline
$\Gamma_2'$         & 0 & $\gamma_2'$ & 0 & 1 & \dots & 0 \\ \hline
$\vdots$            & $\vdots$ & $\vdots$ & $\vdots$ & $\vdots$ & $\ddots$ & $\vdots$ \\ \hline
$\Gamma_m$          & $\gamma_m$ & 0 & 0 & 0 & \dots & 0 \\ \hline
$\Gamma_m'$         & 0 & $\gamma_m'$ & 0 & 0 & \dots & 1 
\end{tabular}}~
\subfloat[][Alice's choice gadget. The variable $x_i$ appears positively in $C_j$ and $C_k$, and negatively in $C_\ell$.]{
        \centering
\begin{tabular}{c||c|c|c|c|c|c}
                    & A         & $U_i$ & $\overline{U_i}$ & $S_j$ & $S_k$ & $S_\ell$ \\ \hline
$X_i$               & 0         & $a_i$ & 0 & 0 & 0 & 0 \\
$\overline{X_i}$    & 0         & 0 & $a_i$ & 0 & 0 & 0 \\ \hline    
$X_{i,1}$           & 0         & $b_i$ & 0 & 1 & 0 & 0 \\
$\overline{X_{i,1}}$& 0         & 0 & $b_i$ & 0 & 0 & 1 \\ \hline
$X_{i,2}$           & 0         & $c_i$ & 0 & 0 & 1 & 0 \\
$\overline{X_{i,2}}$& 0         & 0 & $c_i$ & 0 & 0 & 0 \\ \hline
$T_i^A$             & $t_i^A$   & 0 & 0 & 0 & 0 & 0 \\ 
\end{tabular}
}

\subfloat[][Bob's choice gadget. The variable $y_i$ appears positively in $C_j$ and $C_k$, and negatively in $C_\ell$.]
{
        \centering   
\begin{tabular}{c||c|c|c|c|c|c|c|c}
                    & B         & $V_i$ & $\overline{V_i}$ & $W_i$ & $\overline{W_i}$ & $S_j$ & $S_k$ & $S_\ell$ \\ \hline
$Y_i$               & 0         & $d_i$ & 0 & 0 & 0 & 0 & 0 & 0 \\
$\overline{Y_i}$    & 0         & 0 & $d_i$ & 0 & 0 & 0 & 0 & 0 \\ \hline
$Y_i'$               & 0         & 0 & 0 & $e_i$ & 0 & 0 & 0 & 0 \\
$\overline{Y_i'}$    & 0         & 0 & 0 & 0 & $e_i$ & 0 & 0 & 0 \\ \hline   
$T_i^B$             & $t_i^B$   & 0 & 0     & 0 & 0 & 0 & 0 & 0 \\ \hline
$Y_{i,1}$           & 0         & $f_i$ & 0 & 0 & 0 & 1 & 0 & 0 \\
$\overline{Y_{i,1}}$& 0         & 0 & $f_i$ & 0 & 0 & 0 & 0 & 1 \\ \hline
$Y_{i,2}$           & 0         & 0 & 0 & $g_i$ & 0 & 0 & 1 & 0 \\
$\overline{Y_{i,2}}$& 0         & 0 & 0 & 0 & $g_i$ & 0 & 0 & 1 \\ 
\end{tabular}
}
    \caption{The different gadgets used in our reduction. Rows correspond to agents and are written in decreasing order of efficiencies. Columns correspond to tasks. If two agents are separated by a horizontal line, then the top one is always picked before the bottom one.}
    \label{fig:table-pspace-complete}
\end{figure}

\subsection*{Alice's choice gadget}

Let $i \in\interval{1}{n}$. Suppose that $x_i$ appears positively in the clauses $C_j$ and $C_k$, and negatively in $C_\ell$. We introduce two tasks $U_i, \overline{U_i}$ and seven agents $X_i, \overline{X_i}$, $X_{i,1}, \overline{X_{i,1}}$, $X_{i,2}, \overline{X_{i,2}},T^A_i$ defined as follows: 

\begin{itemize}[nolistsep,noitemsep]
    \item $X_i$ has efficiency $a_i$ in task $U_i$.
    \item $\overline{X_i}$ has efficiency $a_i$ in task $\overline{U_i}$.
    \item $X_{i,1}$ has efficiency $b_i$ in task $U_i$ and efficiency $1$ in task $S_j$.
    \item $\overline{X_{i,1}}$ has efficiency $b_i$ in task $\overline{U_i}$ and efficiency $1$ in task $S_\ell$.
    \item $X_{i,2}$ has efficiency $c_i$ in task $U_i$ and efficiency $1$ in task $S_k$.
    \item $\overline{X_{i,2}}$ has efficiency $c_i$ in task $\overline{U_i}$.
    \item $T^A_{i}$ has efficiency $t^A_i$ in task $A$.
\end{itemize}

 \subsection*{Bob's choice Gadget}

Let $i \in\interval{1}{n}$. Suppose that $x_i$ appears positively in the clauses $C_j$ and $C_k$, and negatively in $C_\ell$. We introduce four tasks $V_i, \overline{V_i}, W_i, \overline{W_i}$ and nine agents $Y_i, \overline{Y_i}, Y_i', \overline{Y_i'}, T_i^B,$ $ Y_{i,1}, \overline{Y_{i,1}}, Y_{i,2},\overline{Y_{i,2}}$ defined as follows: 

\begin{itemize}[nolistsep,noitemsep]
    \item $Y_i$ has efficiency $d_i$ in task $V_i$.
    \item $\overline{Y_i}$ has efficiency $d_i$ in task $\overline{V_i}$.
    \item $Y'_i$ has efficiency $e_i$ in task $W_i$.
    \item $\overline{Y_i'}$ has efficiency $e_i$ in task $\overline{W_i}$.
    \item $T^B_{i}$ has efficiency $t^B_i$ in task $B$.
    \item $Y_{i,1}$ has efficiency $f_i$ in task $V_i$ and efficiency $1$ in task $S_j$.
    \item $\overline{Y_{i,1}}$ has efficiency $f_i$ in task $\overline{V_i}$ and efficiency $1$ in task $S_\ell$.
    \item $Y_{i,2}$ has efficiency $g_i$ in task $W_i$ and efficiency $1$ in task $S_k$.
    \item $\overline{Y_{i,2}}$ has efficiency $g_i$ in task $\overline{W_i}$.
\end{itemize}

\subsection*{Efficiencies and order of the moves}

We write $u \gg v$ to set $u = 5v$. More generally, we write $u_0 \gg u_1 \gg \dots \gg u_k$ to set the efficiencies $u_0 = 5u_1 = \ldots = 5^ku_k $. We define:
$$\alpha \gg \beta \gg \gamma_1 \gg \gamma_1' \gg \gamma_2 \gg \gamma_2' \gg \dots \gg \gamma_m \gg \gamma_m' \gg a_1, $$
and for all $i \in \interval{1}{n}$,
$$ a_i \gg b_i \gg c_i \gg t_i^A \gg d_i \gg e_i \gg t_i^B \gg f_i \gg g_i \gg a_{i+1}, $$
where $a_{n+1}=1$ (note that $a_{n+1}$ does not appear in the game, but it is used to set all the other efficiencies).

Consider an intermediate position where the maximum efficiency over all free agents is $5^k$ for some $k \geq 1$, such that the agents already picked by Alice and Bob all have efficiency 0 in the tasks attaining the value $5^k$. Note that, for all $i \in \interval{1}{k}$, there are at most two free agents whose maximum efficiency is exactly $5^i$. Moreover, we have $ 5^k > 5^k - 5 = 2 \sum^{k-1}_{i=1} 2{\cdot}5^i.$
Therefore:
\begin{itemize}
\item If there is exactly one agent with maximum efficiency $5^k$, then that agent is a dominating agent, so by Lemma \ref{dominating agent} it is optimal for the next player to pick that agent. This will be the case for the agents $A_1, B_1, \Gamma_j, \Gamma'_j, T_i^A, T_i^B$.
\item If there are exactly two agents with maximum efficiency $5^k$, then they form a pair of dominating agents (indeed, when one of these agents is picked, the other becomes a dominating agent in the resulting position), so by Lemma \ref{lemma: paired forcing move} it is optimal for the next player to pick one of the two and for his opponent to immediately pick the other. This will be the case for the pairs of agents $\{X_i, \overline{X_i}\}$, $\{X_{i,1}, \overline{X_{i,1}}\}$, $\{X_{i,2}, \overline{X_{i,2}}\}$, $\{Y_i, \overline{Y_i}\}$, $\{Y_i', \overline{Y_i'}\}$, $\{Y_{i,1}, \overline{Y_{i,1}}\}$, $\{Y_{i,2},\overline{Y_{i,2}}\}$.
\end{itemize}

This leads us to the following move order which is optimal:

\begin{enumerate}[label=(\Alph*)]
    \item Alice picks $A_1$.
    \item Bob picks $B_1$.
    \item For all $j \in \interval{1}{m}$ in increasing order:
    \begin{enumerate}[label=($j$.\arabic*)]
        \item Alice picks $\Gamma_j$.
        \item Bob picks $\Gamma_j'$.
    \end{enumerate} 
    \item For all $i \in \interval{1}{n}$ in increasing order:
    \begin{enumerate}[label=($i$.\arabic*)]
        \item Alice picks one of $X_{i}, \overline{X_{i}}$.
        \item Bob picks the available agent among $X_{i}$ and $ \overline{X_{i}}$.
        \item Alice picks one of $X_{i,1}, \overline{X_{i,1}}$.
        \item Bob picks the available agent among $X_{i,1}$ and $ \overline{X_{i,1}}$.
        \item Alice picks one of $X_{i,2}, \overline{X_{i,2}}$.
        \item Bob picks the available agent among $X_{i,2}$ and $ \overline{X_{i,2}}$.
        \item Alice picks $T^A_{i}$.
        \item Bob picks one of $Y_{i}, \overline{Y_{i}}$.
        \item Alice picks the available agent among $Y_{i}$ and $ \overline{Y_{i}}$.    
        \item Bob picks one of $Y'_{i}, \overline{Y'_{i}}$.
        \item Alice picks the available agent among $Y'_{i}$ and $ \overline{Y'_{i}}$.
        \item Bob picks $T^B_{i}$.
        \item Alice picks one of $Y_{i,1}, \overline{Y_{i,1}}$.
        \item Bob picks the available agent among $Y_{i,1}$ and $ \overline{Y_{i,1}}$.
        \item Alice picks one of $Y_{i,2}, \overline{Y_{i,2}}$.
        \item Bob picks the available agent among $Y_{i,2}$ and $ \overline{Y_{i,2}}$.
    \end{enumerate}
\end{enumerate}

Note that all the moves until move $C.m.2$ are forced moves. Moreover, during each step $C.j.2$ for $j \in \interval{1}{m}$, Bob can already assign Agent $\Gamma'_j$ to Task $S_j$ since Agent $B_1$ already takes care of Task $B$. Therefore, after Steps A through C, the provisional score is $\alpha - (\beta +m)$, and strategies need only be detailed during Step D. 

Set $s = \alpha - \beta$. We prove that $\score(\game) \geq s$ if and only if Satisfier has a winning strategy for the QBF game on $\psi$.

\subsection*{Proof of sufficiency}

Suppose first that Satisfier has a winning strategy $\strat$ for the QBF game on $\psi$. We consider the following strategy for Alice:

\begin{itemize}
    \item If $\strat$ dictates to put $x_i$ to True, she picks $X_i, X_{i,1}$ and $X_{i,2}$. Otherwise, she picks $\overline{X_i}, \overline{X_{i,1}}$ and $\overline{X_{i,2}}$. Note that all three picks force Bob's answer, so this is always possible.
    \item If Bob picks at least one of $Y_i$ or $Y'_i$, then Alice considers in $\strat$ that Falsifier has put $y_i$ to False. Otherwise, he has picked both $\overline{Y_i}$ and $\overline{Y_i'}$, and she considers that Falsifier has put $y_i$ to True. Next, if Bob has picked $Y_i$, then she picks $Y_{i,1}$, otherwise she picks $\overline{Y_{i,1}}$. Finally, if Bob has picked $Y_i'$, then she picks $Y_{i,2}$, otherwise she picks $\overline{Y_{i,2}}$.
\end{itemize}

Following this strategy, in Alice's choice gadget, both Alice and Bob will score $a_i$ in either Task $U_i$ or Task $\overline{U_i}$ and will not score in the other. In Bob's choice gadget, both Alice and Bob will score $d_i$ in either Task $V_i$ or Task $\overline{V_i}$ and will not score in the other, moreover both Alice and Bob will score $e_i$ in either Task $W_i$ or Task $\overline{W_i}$ and will not score in the other.

Therefore, during Step $D$, only Alice can score points, as Bob had already scored 1 in all the tasks $S_1,\ldots,S_m$ thanks to the setup gadget. We now prove that Alice scores 1 in each of the tasks $S_1,\ldots,S_m$.

Let $j \in \interval{1}{m}$. Since $\strat$ is a winning strategy for Satisfier, there exists a variable literal $\ell_i$ in the clause $C_j$ that is satisfied.

\begin{itemize}
    \item If $\ell_i = x_i$ then, by definition of Alice's strategy, she has picked $X_i, X_{i,1}$ and $X_{i,2}$. 
    Since Agent $X_i$ should obviously be assigned to Task $U_i$ (it is its only nonzero efficiency, and it is higher than the efficiencies of $X_{i,1}$ and $X_{i,2}$), whichever of Agents $X_{i,1}$ or $X_{i,2}$ has efficiency 1 for Task $S_j$ should be assigned to it.
    \item If $\ell_i = \neg x_i$ then, by definition of Alice's strategy, she has picked $\overline{X_i}, \overline{X_{i,1}}$ and $\overline{X_{i,2}}$. 
    Since Agent $\overline{X_i}$ should obviously be assigned to Task $\overline{U_i}$ (it is its only nonzero efficiency, and it is higher than the efficiencies of $\overline{X_{i,1}}$ and $\overline{X_{i,2}}$), she should assign Agent $\overline{X_{i,1}}$ to Task $S_j$ in which it has efficiency 1.
    \item If $\ell_i = y_i$, then Alice has considered that Falsifier has put $y_i$ to True in $\psi$. This only happens when Bob has picked both $\overline{Y_i}$ and $\overline{Y_i'}$. Thus, Alice has picked $Y_i, Y_i', Y_{i,1}$ and $Y_{i,2}$.
    Since Agents $Y_i$ and $Y'_i$ should obviously be assigned to Tasks $V_i$ and $W_i$ respectively, Alice should assign whichever of Agents $Y_{i,1}$ and $Y_{i,2}$ has efficiency 1 for Task $S_j$ to that task.
    \item If $\ell_i = \neg y_i$, then Alice has considered that Falsifier has put $y_i$ to False in $\psi$. Thus, Bob has picked at least one of $Y_i$ or $Y'_i$. If Bob has picked $Y_i$, then Alice has picked $\overline{Y_i}$ and $\overline{Y_{i,1}}$: since $\overline{Y_i}$ should obviously be assigned to Task $\overline{V_i}$, $\overline{Y_{i,1}}$ should be assigned to Task $S_j$ in which it has efficiency 1.
    If Bob has picked $Y'_i$, then Alice has picked $\overline{Y'_i}$ and $\overline{Y_{i,2}}$: since $\overline{Y'_i}$ should obviously be assigned to Task $\overline{W_i}$, $\overline{Y_{i,2}}$ should be assigned to Task $S_j$ in which it has efficiency 1.
\end{itemize}

We see that, in all cases, Alice has also scored 1 in each of the tasks $S_1, \dots S_m$. Therefore, the final score is $\alpha - \beta = s$.

\subsection*{Proof of necessity}

Suppose first that Falsifier has a winning strategy $\strat$ for the QBF game on $\psi$. We consider the following strategy for Bob:

\begin{itemize}
    \item If Alice picks at least two agents in $\{X_i, X_{i,1}, X_{i,2}\}$, then Bob considers for $\strat$ that Satisfier has put $x_i$ to True. Otherwise, he considers that Satisfier has put $x_i$ to False.
    \item If $\strat$ dictates to put $y_i$ to True, then Bob picks $\overline{Y_i}$ and $\overline{Y_i'}$. Otherwise, he picks $Y_i$ and $Y_i'$. 
\end{itemize}

Remark $(\ast)$: The first time Alice picks an agent with a nonzero efficiency in Task $U_i, \overline{U_i}, V_i, \overline{V_i}, W_i,\overline{W_i}$ respectively, we can assume that she assigns that agent to that task, and that the agent picked by Bob as an answer is assigned to Task $\overline{U_i}, U_i, \overline{V_i}, V_i, \overline{W_i},W_i$ respectively.

Indeed, this is immediate for Bob, since he had already scored 1 in each of the tasks $S_1,\ldots,S_m$ earlier thanks to the setup gadget. As for Alice, this comes as a consequence: indeed, if she does not assign her agent to Task $U_i, \overline{U_i}, V_i, \overline{V_i}, W_i,\overline{W_i}$ respectively, then she will score at least five times less points than Bob will in Task $\overline{U_i}, U_i, \overline{V_i}, V_i, \overline{W_i},W_i$ respectively, and the sacrifice is not worth scoring 1 in some Task $S_j$.

Since $\strat$ is a winning strategy for Falsifier, $\varphi$ is not satisfied once all agents are picked. Therefore, there exists $j \in \interval{1}{m}$ such that no literal of the clause $C_j$ is satisfied. We now show that Alice does not score at all in Task $S_j$. Let $\ell_i$ be a literal of $C_j$.

\begin{itemize}
    \item If $\ell_i = x_i$, then Bob has considered that Satisfier has put $x_i$ to False. This only happens when Alice has picked at least two agents in $\{\overline{X_i}, \overline{X_{i,1}}, \overline{X_{i,2}}\}$. Note that all three of these agents have efficiency 0 for Task $S_j$. In particular, if Alice has picked all three, then she cannot score in Task $S_j$ during this step. Therefore, suppose that she has picked exactly two of them, i.e. she has picked exactly one agent in $\{X_i, X_{i,1}, X_{i,2}\}$. 
    By Remark $(\ast)$, Alice will necessarily assign that agent to Task $U_i$. Thus, she does not score in task $S_j$ during Step D.$i$.
    \item If $\ell_i = \neg x_i$, then Bob has considered that Satisfier has put $x_i$ to True. This only happens when Alice has picked at least two agents in $\{X_i, X_{i,1}, X_{i,2}\}$. Again, note that all three of these agents have efficiency 0 for Task $S_j$. In particular, if Alice has picked all three, then she cannot score in Task $S_j$ during this step. Therefore, suppose that she has picked exactly two of them, i.e. she has picked exactly one agent in $\{\overline{X_i}, \overline{X_{i,1}}, \overline{X_{i,2}}\}$. By Remark $(\ast)$, Alice will necessarily assign that agent to Task $\overline{U_i}$. Thus, she does not score in task $S_j$ during Step D.$i$.
    \item If $\ell_i = y_i$ then, by definition of Bob's strategy, he has picked $Y_i$ and $Y_i'$, and Alice has picked $\overline{Y_i}$ and $\overline{Y_i'}$.
    Note that Agents $\overline{Y_{i,1}}$ and $\overline{Y_{i,2}}$ have efficiency 0 for Task $S_j$. Moreover, if Alice has picked $Y_{i,1}$ (resp $Y_{i,2}$), then, by Remark $(\ast)$, she must assign that agent to Task $V_i$ (resp. Task $W_i$). All in all, Alice does not score in Task $S_j$ during Step D.$i$.
    \item If $\ell_i = \neg y_i$ then, by definition of Bob's strategy, he has picked $\overline{Y_i}$ and $\overline{Y_i'}$, and Alice has picked . $Y_i$ and $Y_i'$.
    Note that Agents $Y_{i,1}$ and $Y_{i,2}$ have efficiency 0 for Task $S_j$. Moreover, if Alice has picked $\overline{Y_{i,1}}$ (resp $\overline{Y_{i,2}}$), then, by Remark $(\ast)$, she must assign that agent to Task $\overline{V_i}$ (resp. Task $\overline{W_i}$). All in all, Alice does not score in Task $S_j$ during Step D.$i$.
\end{itemize}

In conclusion, Alice does not score in Task $S_j$. Additionally, by Remark $(\ast)$, Bob matches Alice's score in the union of the tasks $U_i, \overline{U_i}, V_i, \overline{V_i}, W_i,\overline{W_i}$. Therefore, Alice scores at most $m-1$ more points than Bob during Step D, so $\score(\game) \leq \alpha - (\beta +m) + (m-1) = \alpha - \beta - 1 < s$.
 \end{proof}

\section{Efficient algorithms for OTPs}\label{section4}

Since \draft~is \PSPACE-complete for TTPs by Theorem \ref{theo_pspace}, this section is dedicated to the algorithmic study of OTPs. A first remark is that maximal efficiencies agents are always preferable, hence the following lemma as a direct corollary of Lemma~\ref{lemma dominating agents}.

\begin{lemma} \label{OTP_maximal}
    Let $\pos$ be a position of the draft game. There is always an optimal move for the next player which consists in picking an agent which has the largest efficiency of all free agents in its associated task. 
\end{lemma}

We can then simplify some positions of the game, as stated in the following lemma.

\begin{lemma} \label{lemma: OTP alice and bob have an agent on the same task.}
    Let $\pos$ be a position that can be reached under optimal play, in which Alice and Bob have both already picked an agent with a nonzero efficiency in the same task $T$. Then the optimal score would be unchanged if we removed from $\pos$ all the free agents with a nonzero efficiency in $T$.
\end{lemma}

\begin{proof}
    First, by Lemma~\ref{OTP_maximal}, there is at most one agent per task that can be an optimal pick: the one with the highest efficiency. Moreover, once Alice and Bob have picked an agent each with a nonzero efficiency in some task $T$, all the other agents having a nonzero efficiency in Task $T$ have a smaller efficiency than the one already picked. Thus, they will not be assigned to this task in the end, and will not contribute to the score.
    Hence, removing them does not change the optimal score.
 \end{proof}

We describe two algorithms to compute the optimal score for OTPs: one is in linear time for the case where there are only two tasks, the other is in $\XP$~time parameterized by the number of tasks. 

\begin{theorem}\label{linear algorithm 2 tasks OTPS}
    There is a linear-time algorithm to compute the optimal score when all agents are OTPs and only two tasks are to be filled.
\end{theorem}

\begin{proof}
We first sort all the agents, according to the task in which they are OTPs. Up to adding $(0,0)$ agents, we assume that both tasks $T$ and $S$ have the same number of OTPs. Denote by $X_1 = (x_1,0), \dots, X_m = (x_m,0)$ the agents for Task $T$ and by $Y_1 = (0,y_1), \dots, Y_m = (0,y_m)$ the agents for Task $S$, such that $x_1 \ge x_2 \ge \dots \ge x_m$ and $y_1 \ge y_2 \ge \dots \ge y_m$. According to Lemma~\ref{OTP_maximal}, Alice will either pick $X_1$ or $Y_1$ on her first move.
By symmetry, computing the optimal score in linear time when Alice starts by picking $X_1$ is sufficient to prove the theorem. Therefore, we assume that Alice starts by picking $X_1$. Now, if Bob picks $X_2$, then all the other $T$-agents can be removed by Lemma~\ref{lemma: OTP alice and bob have an agent on the same task.}, so optimal play continues with Alice picking $Y_1$ and Bob picking $Y_2$, resulting in a score of $\alpha = x_1 - x_2 + y_1 - y_2$. Therefore, we assume that play starts with Alice picking $X_1$ and Bob picking $Y_1$.

For all $i \in \interval{1}{m}$, we define:

   $$B(i) = \left\{
    \begin{array}{ll}
         \min \left (\alpha, A(2) \right) & \quad \mbox{if } i = 1 \\
         x_1 - y_1 & \quad \mbox{if } i = m \\
        \min \left (x_1 - x_{i+1} + y_i - y_1, A(i+1) \right) & \quad \mbox{otherwise}
    \end{array}
\right.$$ and 
$$A(i)= \max\left (x_1 - x_i + y_i - y_1, B(i) \right ).$$

Note that the functions $A$ and $B$ can be computed in linear time in decreasing order of $i$, since the induction stops after at most $m$ steps and the complexity satisfies $C(i) = C(i+1) + O(1)$ where $C(m)$ is a constant.

We now claim that the optimal score is exactly $B(1)$. First note that the game ends one turn after the first time that Alice picks an $S$-agent or that Bob picks a $T$-agent since, by Lemma~\ref{lemma: OTP alice and bob have an agent on the same task.}, all the other $S$-agents (resp. $T$-agents) can then be removed. 

Consider Bob's $i$-th turn, with $i \in\interval{1}{m}$. If he picks $X_{i+1}$ 
, then necessarily Alice has picked every $X_j$ for $j\in\interval{1}{i}$ and Bob has picked every $Y_j$ for $j\in\interval{1}{i-1}$. Hence, by Lemma~\ref{lemma: OTP alice and bob have an agent on the same task.}, Alice will pick $Y_{i}$ and the score will be $x_1 - x_{i+1} + y_{i} - y_1$. If he picks $Y_i$ instead, then it is now Alice's turn and, as we are about to show, the score will be $A(i+1)$. Since Bob wants to minimize the score, the optimal score before Bob's turn is exactly the minimum among these two values i.e. $B(i)$. 

Consider Alice's $i$-th turn, with $i \in\interval{2}{m}$. If she picks $Y_i$, then necessarily Bob has picked all the $Y_j$s for  $j\in\interval{1}{i-1}$ and Alice has picked all the $X_j$s for $j\in\interval{1}{i-1}$, so Bob will pick $X_i$ and the score will be $x_1 - x_{i} + y_i - y_1$. If she picks $X_i$ instead, then it is now Bob's turn and and the score will be $B(i)$. Since Alice wants to maximize the score, the optimal score before Alice's turn is exactly the maximum among these two values i.e. $A(i)$.
 \end{proof}

\begin{theorem}\label{xp}
    There is an \XP\ algorithm, parameterized by the number of tasks $t$, that computes the optimal score in time $O(n^t)$ when all agents are OTPs, where $n$ is the number of agents.
\end{theorem}

\begin{proof}
By Lemma~\ref{OTP_maximal}, the optimal score of the game can be computed recursively as follow :
\begin{itemize}
\item In any position where no agent is left to pick, the optimal score of the game is exactly the sum of the best agent for each task for Alice, minus the sum of the best agent for each task for Bob.
\item In any position where Alice is to pick next and there is at least one agent to pick, the optimal score is the maximum score among the children positions for all possible choice of a highest efficiency agent for Alice.
\item In any position where Bob is to pick next and there is at least one agent to pick, the optimal score is the minimum score among the children positions for all possible choice of a highest efficiency agent for Bob.
\end{itemize}

We prove that the number of positions to consider is upperbounded by $O(n^t)$ when $t$ is bounded.

Let $t$ be the number of tasks. Suppose that for $1 \le j \le t$, only $n_j$ agents have a nonzero efficiency in task $j$. Hence, the total number of agents (up to remove the one with only zero efficiencies) is $n = n_1 + \dots  + n_t$. As agents are OTPs, we first sort according to each task and to their efficiency, which can be done in $O\big (t \cdot n \log(n)\big )$.
    
By Lemma~\ref{OTP_maximal}, we can suppose that any moment of the game, the agent that will be played is a free agent with highest efficiency in the task in which he has a nonzero efficiency. Let us remark that in a given position the score of each player can be computed as the sum of the best agent the picked for each task (or $0$ if they didn't pick any). Thus, to evaluate a position, there is no need to remember exactly the agents picked by the players, but only the best agent picked by each player for each task. Moreover, by Lemma~\ref{lemma: OTP alice and bob have an agent on the same task.}, if two players picked at least one agent for a given task, we can then without loss of generality consider that no further picks will be made for that task. Lastly, as the agent with a nonzero efficiency in a given task are always picked in decreasing order of their efficiency,  we only need to keep track of the number of agents that are left for this task.

We denote by $x_i^j$ the $i$-th efficiency of an agent for task $j$ in decreasing order.

Hence, for each task:
\begin{itemize}
\item We say that the task is {\em closed} if both players picked an agent with a nonzero efficiency in this task or if there are no agent with a nonzero efficiency left in it. It is then sufficient to know the identity of the player that picked first for this task, and the index of the agent that was taken by the other player ($0$ if said player did not pick any agent with a nonzero efficiency in this task).
\item Otherwise, we note the identity of the only player that may have played in that column, and the number of agents with nonzero efficiencies in this task that are left to pick.
\end{itemize}

Thus, there are $4$ possible situations for a task (closed with best agent picked by Alice / by Bob, not closed with only Alice / Bob picks) and one value between $0$ and $n_j$ to remember, resulting in at most $4(n_j+1)$ possible states for that task.
Identifying that $(A_c,n_j)$, $(B_c,n_j)$, $(A_c,0)$ and $(B_c,0)$ correspond to no possible position allows to reduce that number to $4 n_j$.
Thus, the number of distinct reduced positions is at most $2\prod_{j=1}^{t}4n_j\leq 2\cdot 4^t\left(\frac{n}{t}\right)^{t}$. 

Then we define a reduced position as an element of: 
$$\quickset{A,B} \times \prod_{j=1}^{t}\left(\quickset{A,B,A_c,B_c}\times\interval{0}{n_j}\right) $$

A reduced position $C=(t,((p_1,q_1), .. , (p_t,q_t)))$ represents the fact that it's the turn of player $t$ (Alice if $t=A$, Bob if $t=B$), and for the $j$-th task :
\begin{itemize}
\item If $p_j=A_c$ (resp $B_c$) then Alice (resp. Bob) picked the best agent for that task, and the best agent picked by Bob (resp. Alice) for that task is the $q_j$-th (the remaining agents to pick for that task are worse than the $q_j$-th).
\item If $p_j=A$ (resp $B$), then Alice (resp. Bob) picked the best agent for that task, Bob (resp. Alice) picked none for that task, and the remaining agents to pick for that task are the $q_j$ worst competent agents for that task.
\end{itemize}

Then, from a position whose reduce position is $C$, if the player $t$ picks the best agent for a task $j\in \interval{1}{t}$, by Lemma~\ref{lemma: OTP alice and bob have an agent on the same task.}, we can suppose $p_j \in \quickset{A,B}$, 
the resulting position has reduced position $C' =(t',((p_1',q_1'), .. , (p_t',q_t')))$ such that:
\begin{itemize}
\item $t'=A$ (resp. $B$) if $t=B$ (resp. $A$),
\item $\forall i \in \interval{1}{t}\setminus\quickset{j}, p_i'=p_i, q_i'=q_i$,
\item \begin{itemize}
	\item If $q_j=n_j$ this was the first agent with a nonzero efficiency picked for task $j$, and we set $p_j'=t$, $q_j'=n_j -1$ 
	\item Else, if $p_j=t$, only the player $t$ has picked agents with nonzero efficiencies in that task. We then set $p_j'=p_j$, and $q_j' = q_j-1$
	\item Otherwise, the player $t$ picks an agent with a nonzero efficiency in a task in which the other player also have an agent with a nonzero efficiency. Thus, we have $p_j = t'$ with $t' \in \{A,B\}$ and $t' \neq t$. If $t=A$ (resp. $B$), we set $p_j'=B_c$ (resp. $A_c$), and $q_j'=q_j$ 
	\end{itemize}
\end{itemize}

Then, the optimal score of $C$, $sc(C)$, which is the optimal score of every position whose reduced position is $C$, can be computed as follows:
\begin{itemize}
\item If for every $j\in\interval{1}{t}$, either $p_j \in \quickset{A_c,B_c}$ or $q_j=0$, then no move will change the best agent for each task for each player. Writing $f(A_c)=0$, $f(B_c)=1$, and $x^j_0 = 0$ we have $sc(C) = \sum_{j=1}^{t}(-1)^{f(p_j)}(x^j_{n_j}-x^j_{q_j})$,
\item Else, if $t=A$ (resp. $B$), $sc(C)$ is the maximum (resp. minimum) of the $sc(C')$ where $C'$ is taken among the resulting reduced positions after Alice (resp. Bob) picked the best agent for a task $j\in \interval{1}{t}$ for which there is at least one player that has picked no competent agent and there are competent agents left.
\end{itemize}

This computation considers all the possible optimal moves and therefore computes the optimal score of the game. By a dynamic programming approach, the total number of values to compute is the number of reduced positions, \textit{i.e.} at most $2 \left(\frac{4}{t}\right)^t n^t$, and each value can be computed in time $O(t)$ knowing the next values in the dynamic programming. Thus, the final complexity is proportional to $t \left(\frac{4}{t}\right)^t n^t$. 
 \end{proof}

\section{Conclusion}\label{section5}

We have introduced a two-player version of the assignment problem, thus initiating the mathematical study of drafts. For an arbitrarily large number of tasks, we have established that this problem is {\sf PSPACE}-complete, even restricted to agents that have at most two nonzero efficiencies. However, we do not know what the complexity is even when there are just two tasks in total. On the other hand, when all agents have at most one nonzero efficiency, we have an {\sf XP} algorithm parameterized by the number of tasks, but the complexity for an arbitrary number of tasks remains unknown: the different tasks surprisingly affect each other in a nontrivial way, but perhaps not sufficiently to allow for {\sf NP}-hardness reductions from what we suspect. Another open question would be to find an explicit strategy for Alice ensuring a nonnegative score, which we only know exists by a theoretical argument.

Some natural modifications could be made to our model. A first alternative could be to allow the value $-\infty$ for efficiencies, to account for the fact that some agents are unable to perform certain tasks (because they are not qualified, or simply because of their nature e.g. humans cannot perform the tasks of machines). Unassigned tasks should then count for $-\infty$ towards the score. It should be noted that the {\sf PSPACE}-completeness result, which is arguably the main result in this paper, still holds in this setting. Indeed, the optimal score of the game is clearly unchanged under replacing every efficiency 0 by $-\infty$ (to adapt the definition of a TTP agent) and then adding two agents per task with efficiency 0 in that task and $-\infty$ everywhere else (to ensure that both players ``score'' at least 0 in each task). A second alternative would be to ask for agents to be assigned a task as soon as they are picked, which would make sense for instance in the context of hiring people for jobs. Again, the {\sf PSPACE}-completeness result would still hold, since there is never a doubt as to which task an agent will eventually be assigned to in our proof.

We have defined the score as the difference between the two players' optimal assignment values, which corresponds to what we may call a {\em Maker-Maker} version of this scoring game. Instead, one may consider the {\em Maker-Breaker} version, where the score is simply defined as Alice's optimal assignment value. In other words, Alice wants to assemble the best team possible, while Bob is not building a team but is instead removing agents to impede Alice. This is a natural variant since it models the assignment problem in an environment where some agents may become unavailable over time, e.g. because they are recruited elsewhere. Instead of two entities competing against each other, this would be seen as a single entity (Alice) trying to protect against the worst-case scenario (incarnated by Bob). We note that the Maker-Breaker scoring game is a Maker-Breaker positional game with hyperedges of size equal to the number of tasks (corresponding to the sets of agents that attain the desired score), therefore the optimal score can be computed in polynomial time if there are three tasks or less \cite{makerbreaker3}.



%
%
%


\bibliographystyle{plain}
\bibliography{biblio}

\end{document}